\documentclass[fleqn,10pt]{wlscirep}
\usepackage[utf8]{inputenc}
\usepackage[T1]{fontenc}
\usepackage{amsmath}
\usepackage{mathtools}
\usepackage{amsthm}
\usepackage{algorithm}
\usepackage{algpseudocode}

\newtheorem{theorem}{Theorem}

\title{Simulating cell populations with explicit cell cycle length - implications to cell cycle dependent tumour therapy}

\author[1,2*]{Peter Boldog}
\author[2,3]{Gergely Röst}
\affil[1]{HUN-REN Wigner Research Centre for Physics, Department of Computational Sciences, Budapest, Hungary}
\affil[2]{Bolyai Institute, University of Szeged, Szeged, Hungary}
\affil[3]{Computational Medicine Group\\ Hungarian Center of Excellence for Molecular Medicine (HCEMM), Szeged, Hungary}

\affil[*]{boldog.peter@wigner.hun-ren.hu}

\begin{abstract}
In this study, we present a stochastic simulation model designed to explicitly incorporate cell cycle length, overcoming limitations associated with classical compartmental models. 
Our approach employs a delay mechanism to represent the cell cycle, allowing the use of arbitrary distributions for cell cycle lengths. 

We demonstrate the feasibility of our model by fitting it to experimental data from melanoma cell lines previously studied by Vittadello et al. 
Notably, our model successfully replicates experimentally observed synchronization phenomena that multi-stage models could not adequately explain. 
By using a gamma distribution to model cell cycle lengths, we achieved excellent agreement between our simulations and empirical data, while significantly reducing computational complexity and parameter estimation challenges inherent in multi-stage approaches. 

Our results highlight the importance of explicitly incorporating cell cycle lengths in modeling cell populations, with potential implications for optimizing cell cycle-dependent tumor therapies.

\end{abstract}

\begin{document}

\flushbottom
\maketitle
%
%
\thispagestyle{empty}

\section*{Introduction}
Cell cycle is a series of precisely controlled events that take place in a cell as it grows and divides.
During the cycle, cells undergo many structural changes, each of which takes a considerable amount of time.
As we will see later, this makes it much more difficult to use classical compartmental models, which assume that the time spent in each compartment is exponentially distributed. 
Division of eukaryotic cells starts with the synthesis of the necessary RNA and protein structures during the so called \textit{G1 phase}. 

Cells go through several \textit{checkpoints} to ensure that the necessary number of molecules are present to start the replication of the genetic material.
The last such checkpoint is the so called \textit{no return point} after which the molecular machinery that is assembled in G1 begins to copy the DNA of the cell. 
This is called the \textit{S phase}, where `S' stands for the synthesis of the nuclear DNA. 
When the two DNA copies are ready for the two daughter cells, the mother cell enters the \textit{G2 phase} in which continues protein and RNA synthesis and approximately doubles its size. 
Then, the mother cell goes through nuclear division (mitosis) and cell division (cytokinesis) producing two daughter cells during the \textit{M phase}.
After M, cells may go into the quiescent \textit{G0 phase}, that may last hours or even the whole lifetime of the cells (depending on its type and the circumstances), or start the whole process over.
For further details we refer to the excellent textbook by Nelson and Cox \cite{lehninger}, the cell cycle is illustrated in Fig. \ref{fig:cell_cycle}.
Given a cell population, cells may be in different stages of the cell cycle.
A population is synchronized, when the cells are in more or less the same state. 
In contrast, when most of the cells are in different state, the population is asynchronous.
Naturally, it may be partially synchronized as well (c.f. Fig. \ref{fig:cell_cycle}).

\begin{figure}[h!]
	\centering 
	\includegraphics[ width=.9\textwidth]{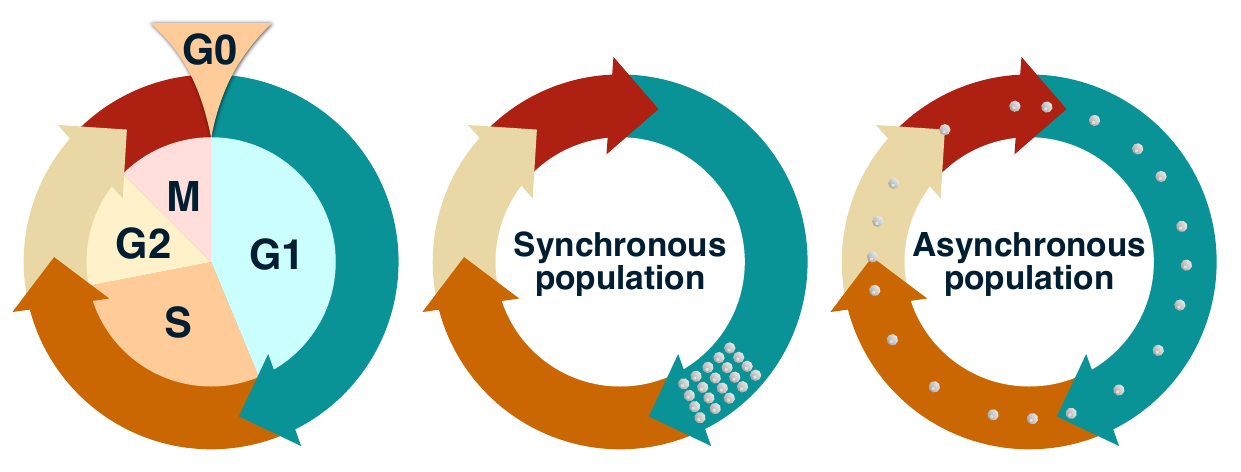}
    \caption{Cell cycle, the series of biochemical events that lead to cell division. 
    The  diagram on the left illustrates the order of the phases in the cycle: G1, S, G2, M. 
    After mitosis cells may enter the quiescent G0 phase. 
    The whole cycle (without G0) last for several hours, typically between 8-24 hours. 
    The diagram in the middle illustrates a synchronous population, when most of the cells are in the same phase in their cycles. 
    The one on the right illustrate an asynchronous population.}
	\label{fig:cell_cycle}
\end{figure}

Cell cycle complicates both mathematical modeling and simulating cell cultures. 
In fact, when cell cycle plays important role in the behavior of the population, the inherent assumption of exponentially distributed cell cycle length in ordinary differenial equations (ODEs) may lead to false conclusions as ``the most probable time for a cell to divide is the current time'' (Yates et al. \cite{yates}).
Each phase of the cycle is of typical length, usually several hours. Thus, when the length of the cell cycle plays an important role, assuming exponentially distributed length is a poor choice.

Moreover, as Vitadello et al. \cite{Vittadello} point out, cell population models based on ODEs inherently assume that the concerned population is asynchronous throughout the modeled period of time.
This may have important implications on various fields, including tissue development or even tumor therapy.
Thus the classical approach, with a system of ODEs, is neither capable of capturing the true dynamics influenced by the length of the cell cycle, nor able to model a synchronized population or spontaneous synchronization.
In the remaining part of the introduction, we summarize some of the main results on the effect of cell cycle in the modeling of cellular populations.

Based on 77 data sets for 16 cell types, Golubev found that the best fitting distribution for the experimentally measured cell cycle lengths was the exponentially modified gamma or exponentially modified Gaussian distribution \cite{Golubev}.
To approximate this, using an integer shape parameter $k$ to keep the Markovian property of the agent based stochastic framework, Yates et al. modeled the cell cycle with an exponentially modified Erlang distribution. 
They suggest a general multi-stage model of the cell cycle that consists of a chain of consecutive states, where the transition rate between $X_i\rightarrow X_{i+1}$ for $i\in\{1,\dots,k-1\}$ is $\lambda_1$, cells spend an exponentially distributed waiting time in each state with mean $\lambda_1^{-1}$.
Finally, cells leave the terminal state of the cycle and split to two daughter cells with rate $\lambda_2$  \cite{yates}:

$$X_1\xrightarrow[]{\lambda_1}X_2\xrightarrow[]{\lambda_1}\dots \xrightarrow[]{\lambda_{1}}X_{k}\xrightarrow[]{\lambda_2}2X_1.$$

Vitadello et al.\cite{Vittadello} extended this approach to explain the behavior of melanoma cell lines C8161, WM983C, and 1205Lu.
In their laboratory experiment they stained cells with a cell cycle-dependent fluorescent dye, called FUCCI.
This dye fluoresces in three different colors according to the different stages of the cell cycle: 
nuclei of cells in G1 appeares red, in case of phase S/G2/M the dye fluoresces green and during G1/S transition the nucleus appears yellow.
Thus, they divided the cell cycle into three main parts, which were named $R,Y,G$ corresponding to the colors. 

In the experiments, the number of cells in each $R(t),Y(t),G(t)$ stage was tracked using time-laps videomicroscopy.
The total number of cells during the experiment is: $M(t)=R(t)+Y(t)+G(t)$.
It was assumed that the lengths of compartment $R,Y,$ and $G$ are $L_R, L_Y, L_G$, respectively, and these are composed of several additional $n$ \textit{auxiliary} stages.
Transitions between these stages are modeled with a chain of ordinary differential equations, where the transition rates are $\lambda_l=n/L_l, l\in\{R,Y,G\}$, and the model consists of the following equations:
\begin{center}
\begin{equation*}
\frac{\mathrm{d}R_i(t)}{\mathrm{d}t} = \left\{
        \begin{array}{ll}
        2\lambda_G G_N(t) - \lambda_R R_1(t) &\quad \text{if } i=1 \\
        \lambda_R R_{i-1}(t) - \lambda_R R_i(t) & \quad \text{if } i=2,\dots N 
    \end{array}
    \right.
\end{equation*}

\begin{equation*}
\frac{\mathrm{d}Y_j(t)}{\mathrm{d}t} = \left\{
        \begin{array}{ll}
        \lambda_R R_N(t) - \lambda_Y Y_1(t) &\quad \text{if } j=1 \\
        \lambda_Y Y_{j-1}(t) - \lambda_Y Y_j(t) & \quad \text{if } j=2,\dots N 
    \end{array}
    \right.
\end{equation*}

\begin{equation*}
\frac{\mathrm{d}G_k(t)}{\mathrm{d}t} = \left\{   \begin{array}{ll}
    \lambda_Y Y_N(t) - \lambda_G G_1(t) &\quad \text{if } k=1 \\
    \lambda_G G_{k-1}(t) - \lambda_G G_k(t) & \quad \text{if } k=2,\dots N. 
    \end{array}
    \right.
\end{equation*}
\end{center}

Based on these considerations, the authors fit the following $3N+3$ length parameter vector (Equation S3 on page 9 of Supplementary Material 1 by Vittadello et al.\cite{Vittadello_sup}):
$[R_1(0), \dots, R_N(0), Y_1(0), \dots, Y_N(0), G_1(0), \dots, G_N(0), L_R, L_Y, L_G].$
They managed to achieve\cite{Vittadello_sup} a remarkably good fit to experimental data at the cost of an extreme large number of intermediate stages (between 30--120, depending on the considered cell line). However, it is important to point out that these auxillary stages may vary between experiments, even in the case of a given cell line. For example: in case of cell line C8161 in one of the experiment (page 12, Eq. S10\cite{Vittadello_sup}) they had to use $N=18$ and in another experiment (page 16, Eq. S15\cite{Vittadello_sup}) they had to use $N=40$.

We may therefore say that multistage models usually require a large number of artifact stages that make it problematic to investigate analytically and even interpret biologically.
Moreover, such models, relying on a Markovian process, do not include the length of the cell cycle as an explicit parameter.

Baker and Röst\cite{baker-rost} followed a different approach. 
To capture the dynamics of a glioblastoma cell population that obeys the so called \textit{grow or go} principle in a finite environment with carrying capacity $K$, they introduced a new delayed logistic equation. 
First, cells were stratified into two subpopulations: \textit{ motile} cells $m$, which can move around the medium and begin proliferation with rate constant $r$ and turn to \textit{proliferating} cells $p$. 
Proliferating cells go through the cell cycle and spend time $\tau$ in a non-motile mode. 
After $\tau$ a $p$ cell split into two daughter cells if there is enough place for the daughter cells, otherwise the cell turns back to mobile mode and aborts cell division. 
The probability at time $t$ for a proliferating cell to find a free place is $(K - p(t) - m(t))/(K)$. 
This is represented in the last term of the first equation of their model:
\begin{align*}
m'(t) &= -rm(t) + rm(t - \tau) + rm(t - \tau) \left( \frac{K - p(t) - m(t)}{K} \right), \\
p'(t) &= rm(t) - rm(t - \tau).
\end{align*} 
The authors employ persistence theory, comparison principles, and $L2$-perturbation techniques to demonstrate that all feasible non-trivial solutions converge to a positive equilibrium. 
They also construct a unique heteroclinic orbit using local invariant manifolds, proving it forms the global attractor alongside the equilibria. 

In Baker, Boldog and Röst\cite{ECMI}, the authors present a slightly different mean-field model for the same purposes.
The model is expressed as a system of delay differential equations in which \textit{motile} cells $m$ enter the cell cycle and turn to \textit{proliferating} type $p$ only if they find a free place to occupy for the daughter cell. The probability of this event is $(K - m(t) - 2p(t))/K$, where the number of proliferating cells and reserved space is $2p(t)$. 
After time $\tau$, upon finishing the cell cycle, proliferating cells split to two motile daughter cells:
\begin{align*}
m'(t) &= -rm(t) \frac{K - m(t) - 2p(t)}{K} + 2r m(t - \tau) \frac{K - m(t - \tau) - 2p(t - \tau)}{K}, \\
p'(t) &= rm(t) \frac{K - m(t) - 2p(t) }{K} - rm(t - \tau) \frac{K - m(t - \tau) - 2p(t - \tau)}{K}.
\end{align*}
It is proved\cite{ECMI} that, starting from biologically feasible initial conditions, the solutions converge to a state where all environmental capacity $K$ is occupied by motile cells.
Depending on initial conditions and parameter values, numerical simulations illustrate different growth behaviors including logistic growth and step-function-like growth corresponding to synchrosination of cells.

\section*{Methods}
Our work is highly motivated by the \textit{transition probability model} (TPM) by Smith and Martin \cite{TPM}.
They observed that phase S and G2 are of characteristic length for the cell type and do not show significant variation in a given population. 
Thus, it is legitimate to think that the length of these phases is deterministic. 
On the other hand, the length of phase G1 shows high variability even in a homogeneous population.
According to their TPM model, the life of the cell can be divided into two states, a proliferating state that lasts for a deterministic period of time (consists of phase S, G2 and M), and an inter division phase that we shall call motile state (consists of phase G0 and G1) that lasts for a random period of time, but the probability of leaving this state is always constant.
The latter assumption about the constant transition rate implies that the length of the motile state is exponentially distributed \cite{TPM}.

\subsection*{The model}
Suppose that cells of uniform size are placed in a well mixed liquid medium with carrying capacity $K$.
The cell population is stratified to \textit{resting} or \textit{quiscent} $Q$ and \textit{proliferating} $P$ cells.
At time $t$ the number of cells in these subpopulations are $Q=Q(t)$ and $P=P(t)$, respectively.
Resting cells commit to division and switch to the proliferating phenotype in a contact inhibited reaction with proliferation rate constant $r$.
Then $\vartheta$ time later, upon the cell cycle elapsed, $P$ cells divide and switch back to the motile phenotype.
To keep the contact inhibited nature of the cell division (for detailed discussion of this reaction type see \cite{boldog}), we assume that cells switch phenotype only if they are big enough (formally, cells managed to reserve space for the daughter cell).
With carrying capacity $K$, the number of free spaces at time $t$ is $K-Q(t)-2P(t)$, where $2P(t)$ stands for the actual number of proliferating cells $P(t)$ and the exact same number of reserved spaces.

Assume that for the initial number of cells $Q(0), P(0)$ in the population we have $Q(0)+P(0)<K$. Based on the above considerations we have two reactions: the phenotype switch $Q\rightarrow P$, that is a random contact inhibited reaction with rate constant $r/K$ between the $Q$ cells and the virtual $K-Q-2P$ vacant spaces, and the scheduled cell divisions $P\rightarrow 2Q$.
The propensity of the phenotype switching and the distribution of the corresponding waiting time is:
\begin{equation}\label{eq:wm_cell_cycle}
a = \frac{r}{K}r(K-Q-2P), \hspace{1cm} \tau\sim Exp(a).    
\end{equation}
The deterministic cell cycle is of length $\vartheta$, which may be constant, such as in the mean field model\cite{ECMI} or chosen from any biologically feasible distribution, which may even be empirical.
Suppose that in a given state we have $P$ number of proliferating cells and their division reactions are scheduled for time $\Theta_1,\dots,\Theta_{P}$.
From which the next delayed reaction will occur at time $t_{R_D}$, where 
\begin{equation}\label{eq:cc_tRD}
    t_{R_D} = \min\{\Theta_i :i=1,\dots,P\}.
\end{equation}
Now, we only have to define how to chose the next reaction to execute. At every step of the simulation, we evaluate $\tau$ and $t_{R_D}$. In case $\tau<t_{R_D}$ a quiscent cell enters cell cycle and whenever $t_{R_D}<\tau$ a proliferating cell completes cell cycle.
Thus, we may summarize the transition between the states with the following flow chart:
$Q\xrightarrow[]{\tau}P\xRightarrow[]{\text{ }\vartheta \text{ }}2Q,$
where $\Rightarrow$ stands for the delayed reaction.

\subsection*{Logistic growth with explicit cell cycle length SSA}
The required data structure is an integer valued variable, which we shall call $Q$, to store the number of quiscent cells; and a list $\underbar{P}$ to store the the scheduled times of the division reactions.
The algorithm is presented in the Appendix section.

\subsection*{The model output}
The output of the model is a time series of the stochastic realisation of the jump process.
During the following experiments, the model is run multiple times and the ensemble average (the expectation of the number of cells in subpopulation $Q$ and $P$) and other statistical measures are calculated.
Based on some elementary considerations we may formulate the following theorem about the asymptotic behaviour of the total number of cells $N(t)=P(t)+Q(t)$.

\begin{theorem}[On the asymptotic bahaviour of the model]
Given the model presented in this section, for the time series of the subpopulations and the total population generated with Algorithm 1 it is true that $N(t)$ is a monotonically non-decreasing function of time,  $\lim_{t\rightarrow\infty}Q(t)=K$ and $\lim_{t\rightarrow\infty}P(t)=0$.
\end{theorem}

\begin{proof}
Since the model does not incorporate cell death, it follows naturally that the total population size cannot decrease over time. Quiescent cells commit to cell division with propensity $a = \frac{r}{K}(K - Q - 2P)$, as defined in Eq. \ref{eq:wm_cell_cycle}. Given our assumption that initially \(Q(0) + P(0) < K\), the propensity \(a\) —and thus the probability of cell division— is strictly positive whenever \(Q + 2P < K\), and zero only when \(Q + 2P = K\). Consequently, the total population size \(N(t)\) is a monotonically non-decreasing function of time.

Additionally, by the construction of the algorithm, proliferating (\(P\)) cells deterministically revert to the quiescent (\(Q\)) state after a fixed time interval. Therefore, when the condition \(Q + 2P = K\) is reached, no additional \(Q\) cells initiate division, and existing \(P\) cells eventually transition back to phenotype \(Q\). Thus, $\lim_{t\rightarrow\infty}Q(t)=K$ and $\lim_{t\rightarrow\infty}P(t)=0$.
\end{proof}

\section*{Results}
\subsection*{Comparison with the mean field logistic growth}
In this section, we aim to show the impact of the cell cycle on the population dynamics. 
To do so, we compare our model to the classical logistic model. 
This comparison reveals that our model can better capture the complex dynamics of cell populations, such as the presence of oscillations.
With proliferation rate $r_{mf}$ and the capacity $K$ of the environment, the deterministic logistic growth model obtains the time evolution of the number of cells $N_{mf}$. 
The well-known Verhulst model is as follows:
$$N'_{mf}=\frac{r_{mf}}{K}N_{mf}(K-N_{mf}).$$ 
For the sake of comparison the capacity in the mean-field model and in our approach is chosen to be the same value.
Note that in the logistic model it is inherently assumed that the inter-division time is exponentially distributed and the cell cycle takes no time to complete: upon deciding to divide, cells split up in an instantenious fashion. 
From probability theory it is well known that the expected value of the interdivision time is $r_{mf}^{-1}$.
Thus, we may expect that the mean-field logistic model gives a good approximation to the averaged output of the delayed logistic growth in case the length of the cell cycle is small compared to the inter division time: $\vartheta\ll r_{mf}^{-1}$.

To illustrate this, on Fig. \ref{fig:cc_log} we fixed $r=1$ and $K=10^4$ and chose initial condition $Q(0)=500, P(0) = 0$. 
The total number of cells is $N(t)=P(t)+Q(t)$. 
Every subfigure of Fig.\ref{fig:cc_log} shows the average of 20 realizations of the stochastic model with varying values of constant cell cycle delay $\vartheta=0.01, 1$, and $5$. 

\begin{figure}[h]
	\centering 
	\includegraphics[ width=.99\textwidth]{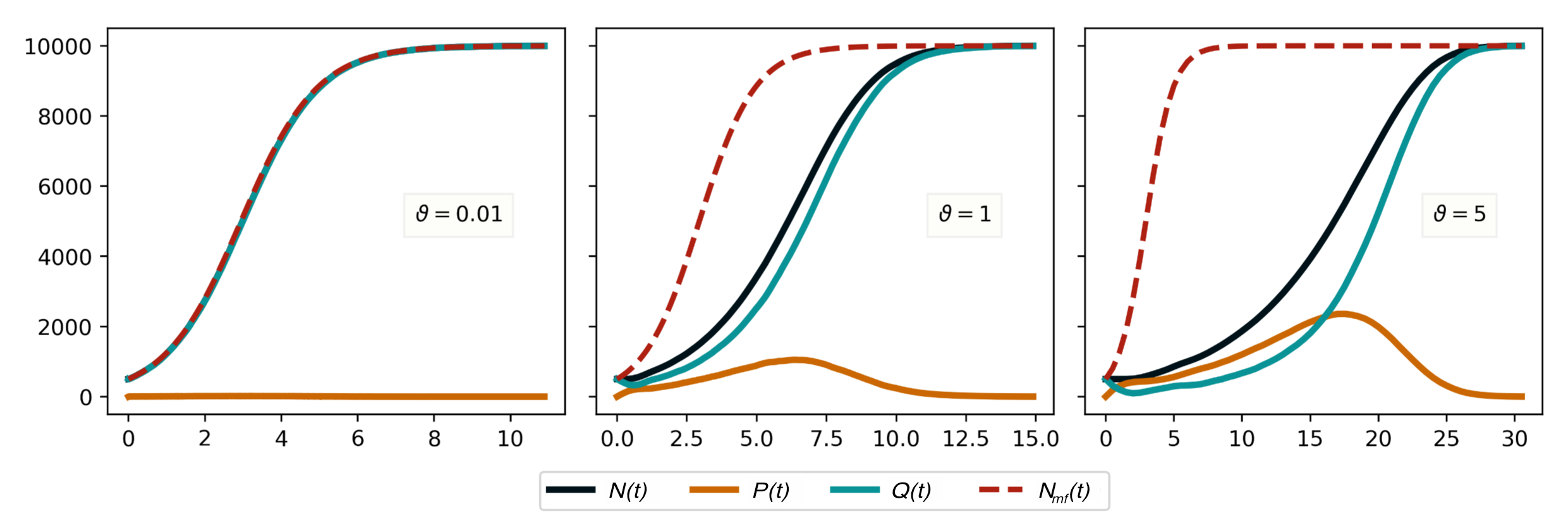}
    \caption{Logistic growth with explicit cell cycle length. 
    The figure illustrates the effect of cell cycle on the dynamics of the population and compares the stochastic delay model with the classical mean field logistic model. 
    Each figure shows the averaged output of 20 independent runs with $r=1, K=10^4$ and initial condition $Q(0)=500, P(0)=0$. 
    The delay parameter, the length of the cell cycle varies on the figure: from left to right the value of $\vartheta$ is $0.01,1$, and $5$, respectively. 
    The legend on the bottom applies for all figures, $N_{mf}(t)$ stands for the mean field logistic curve.}
	\label{fig:cc_log}
\end{figure}

We may obtain a really good agreement between the logistic model and the delayed stochastic model in case of a short cell cycle $(\vartheta=0.01\ll 1=1/r_{mf})$ and it can be seen that almost all cells are in the inter division resting state $Q$ state, as it is expected.
For $\vartheta=1=1/r_{mf}$ the agreement between the logistic model and the stochastic simulation with cell cycle is very poor. 
The delay, caused by the cycle slows the development of the population shifting the curves to the right. 
It also creates a transient oscillation that we will investigate in the next section.
We may observe that a significant amount of proliferating cells appear due to the large cell cycle causing the curve of the total number of cells $N(t)$ and the curve of the number of quiscent cells $Q(t)$ not to coincide.
Finally, for $\vartheta=5$, the non-monotonic behavior of the $P(t)$ curve becomes exaggerated making the transient oscillation in the beginning last longer.

\subsection*{Synchronicity}
The time cells spend in resting mode is the inverse of the phenotype switching rate: $r^{-1}$.
Synchronicity may occur when $r^{-1}$ is small compared to the length of its cell cycle: $r^{-1}\ll\vartheta$. 
This causes the cells to enter the proliferating state more or less together resulting the unusual step-like growth dynamics that can be seen on the left part of Fig. \ref{fig:cc_sync1}. 
This special dynamics may last until the population reaches the capacity $K$ in case of certain parameter configuration.
We will spend more time investigating this in the very next section.
Remarkably, cells in this model are entirely independent and do not influence each other's behavior. 
Thus, this synchronicity is different from the synchronization of coupled oscillators.

\begin{figure}[h!]
	\centering 
\includegraphics[width=.99\textwidth]{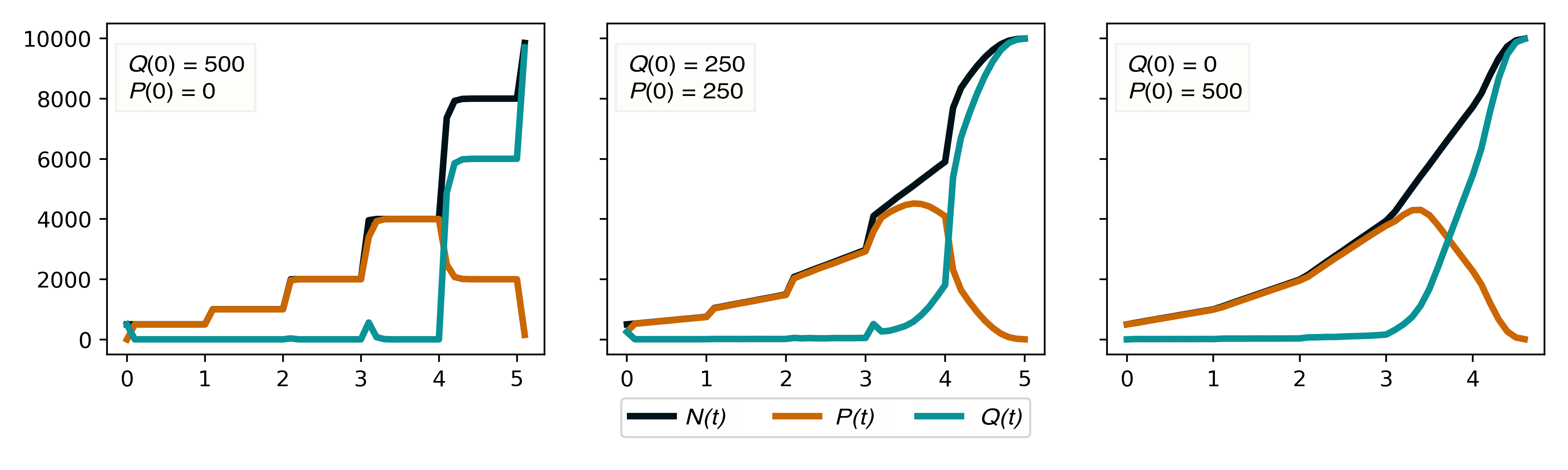}
    \caption{Step-wise synchronization of cells may last for a long time, but it is sensitive for the initial conditions. 
    The figures show that the characteristic step-like dynamics may cease to exist if the initial cell population is in the proliferating state and their scheduled division events are scattered uniformly in the interval $[0,\vartheta]$.}
	\label{fig:cc_sync1}
\end{figure}

Synchronization is very sensitive for initial conditions: on Fig. \ref{fig:cc_sync1} we set $\vartheta=1$ and $r=100$ for all figures, then we plotted the averaged time series of 20 runs for 3 different initial conditions.
In case of the leftmost figure all cells are quiscent initially, $Q(0)=500, P(0)=0$. 
The high rate forces the cells to enter the proliferating state during a very short time interval. Then, the cells spend as long as $\vartheta=1$ time in state P, resulting a horizontal line of length about $\vartheta$. 
Upon the cell cycle elapsed, the population doubles its size in a short time.

On the middle plot of Fig. \ref{fig:cc_sync1} the number of motile and proliferating cells are the same: $M(0)=P(0)=250$.
As cells may be in infinitely many stages in their cell cycle we have to prescribe initial conditions for the cells in state $P$.
We chose the scheduled times $\Theta$ for the delays from continuous uniform distribution $\Theta\sim\text{Uniform}(0,\vartheta)$ as initial condition. 
It is well known that on any finite interval $[a,b]$ among all the continuous distributions, for which the support is this interval, $\text{Uniform}(a,b)$ has the maximum entropy. 
Since the scheduled delay reactions are scattered uniformly in the interval $[0,\vartheta]$ the time between them is approximately the same, $\vartheta/250$. 
Making the cells to reenter state M with a constant rate, causing a slope in the straight line segments.
Finally, in the middle plot of Fig. \ref{fig:cc_sync1}, all cells are proliferating, and their scheduled division times are uniformly distributed within the interval \([0,\vartheta]\), thus eliminating synchronization.

\subsection*{Applying the model to experimental data}
\subsubsection*{Experimental data}
Now, we apply our model to the experimental data provided by Vitadello et al. in the supplementary material of their article \textit{Mathematical models incorporating a multi-stage cell cycle replicate normally-hidden inherent synchronization in
cell proliferation} \cite{Vittadello_sup}.

The authors run several \textit{in vitro} laboratory experiments on three melanoma cell lines: C8161, WM983C, and 1205Lu.
The cells were transfected to express a cell cycle dependent fluorescent dye, FUCCI.
This dye labels cells according to their state in the cell cycle.
The nucleus of a cell in phase G1 fluoresces red, it fluoresces green in phase S, G2, and  M, and in the early phase of S both colors appear and the nucleus fluoresces yellow.
The population was monitored for 48 hours and photographed every 15 minutes under microscope.
Then, from each image, the number of cells were counted and a multi-stage mathematical model were fit to the time series.
The authors claim that they applied a careful preparation on the cell culture to prevent any induced synchronization.
For further details we refer to the supplementary material of \cite{Vittadello_sup}.

\subsubsection*{Fitting to the data}

According to the multi-stage model described by Vittadello et al. \cite{Vittadello}, the cell cycle length follows an exponentially modified Erlang distribution, where the shape parameter \(k\) corresponds to the number of stages. 
However, since our model uses a delay-based approach to represent the cell cycle, it is not necessary for the shape parameter to be restricted to integer values. 
Therefore, following Golubev \cite{Golubev} we assume the cell cycle length \(\vartheta\) follows a gamma distribution \(\vartheta \sim \text{Gamma}(\alpha, \lambda)\), with shape parameter \(\alpha\) and rate parameter \(\lambda\) to be fitted to experimental data.

The initial conditions for the number of motile cells \(Q(0)\) and proliferating cells \(P(0)\) are provided by the experimental setup. However, the initial distribution of proliferating cells—specifically, the distribution of their remaining cell cycle times—must be determined. We assume this distribution also follows a gamma distribution, characterized by parameters \(\alpha_0\) and \(\lambda_0\).
Additionally, the proliferation rate \(r\) must be determined, while the carrying capacity \(K\) is fixed at \(10^5\).

Consequently, the parameter vector to be estimated is \((r, \alpha, \lambda, \alpha_0, \lambda_0)\). To estimate these parameters, we employed the Nelder-Mead simplex optimization method implemented in the SciPy Python library, minimizing the squared differences between the simulated and observed time-series data.
Parameter fitting was performed using data from a single experiment on the C8161 cell line (Fig. \ref{fig:cc_fit}). As demonstrated in the figure, our model provides an excellent fit to the observed data.

\begin{figure}[h!]
	\centering 
	\includegraphics[width=.45\textwidth]{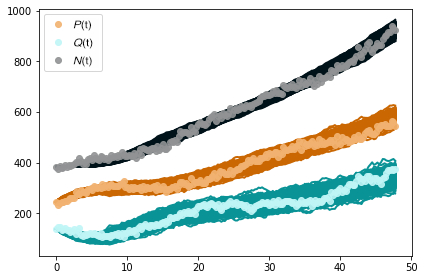}
    \includegraphics[width=.45\textwidth]{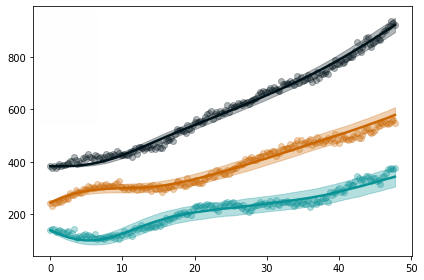}
    \caption{\textbf{Cell cycle model with delay fitted to experimental data by Vittadello et al. \cite{Vittadello_sup}} Left panel: Scatter plots represent measured data points; solid lines show the time series generated from 1000 simulations with fitted parameters.
    Right panel: Simulation results for the C8161 phenotype, showing the empirical mean of 1000 independent runs (solid lines), shaded areas representing ±2 standard deviations around the mean, and experimentally measured data points (scatter plot).}
	\label{fig:cc_fit}
\end{figure}

\section*{Discussion}
Our model presents several advantages compared to conventional multi-stage models. Firstly, by incorporating explicit delays instead of multiple intermediate stages, it significantly simplifies both computational implementation and the interpretation of results. The explicit representation of the cell cycle length allows for precise modeling using biologically realistic distributions, such as the gamma distribution, effectively capturing variability observed experimentally. This flexibility is critical for accurately modeling biological systems where cell cycle timing directly influences population dynamics.

Secondly, the reduced number of parameters in our model compared to multi-stage models simplifies the fitting process. Unlike multi-stage models, where parameter estimation is complicated by numerous artificial intermediate stages, our approach employs fewer biologically interpretable parameters, thus enhancing model reliability and generalizability across different experimental settings.

Moreover, the computational efficiency of our model allows for rapid simulations, facilitating extensive parameter exploration and statistical analyses. This efficiency is particularly beneficial in therapeutic contexts, where understanding population-level responses to cell cycle-dependent treatments requires numerous simulations across diverse scenarios.

Finally, the model effectively captures synchronization phenomena without necessitating direct interactions or coupling between cells, offering valuable insights into intrinsic synchronization mechanisms. This has significant implications for developing and optimizing therapeutic strategies targeting synchronized cell populations in cancer treatments. Overall, our model provides a robust and practical framework for simulating cell population dynamics, emphasizing the critical role of explicit cell cycle length modeling.

\bibliography{bibliography}

\section*{Acknowledgements}

The authors were supported by the National Research, Development and Innovation Office (NKFIH) in Hungary, grants KKP 129877, 2022-2.1.1-NL-2022-00005, TKP2021-NVA-09.

\appendix

\begin{algorithm}\label{alg-1}
\caption{Logistic growth with explicit cell cycle length SSA}\label{alg:Log_CCD_SSA}
\begin{algorithmic}
\State \textbf{Input:} initial number of motile cells $Q_0=Q(0)$ and proliferating cells $P_0=P(0)$, initial distribution of remaining times $(\Theta_1,\dots,\Theta_{P_0})$, proliferating rate constant $r$, carrying capacity $K$, distribution for the length of the cell cycle and halting condition $t_{end}$.\\
\State 
\textbf{Initialisation:} \\
Set $t \gets 0$ \Comment{Set system time to zero.}
\State Set $Q\gets Q(0)$ and $\underbar{P}\gets(\Theta_1,\dots,\Theta_{P_0})$
\Comment{Set initial state of the system.}\\

\While{$ t\leq t_{end} \textrm{ and } 0 < Q+2\cdot len(\underbar  P) < K$}
\State Calculate the propensity function $a$ according to Eq. (\ref{eq:wm_cell_cycle}).
\State Sample $\tau \sim Exp(a)$ according to Eq. (\ref{eq:wm_cell_cycle}).
\State Calculate the minimal of remaining times $t_{R_D}$ according to Eq. (\ref{eq:cc_tRD}).\\
\If{$\tau < t_{R_D}$}  \Comment{\textbf{A quiscent cell enters cell cycle}}
\State $t \gets t + \tau$
\Comment{Advance system time}
\State $\underbar P\gets (\Theta_1-\tau,\dots,\Theta_{P}-\tau)$
\Comment{Decrease remaining times of proliferating cells}
\State $Q\gets Q-1$
\Comment{Decrease the number of motile cells}
\State Obtain $\vartheta$
\Comment{Obtain cell cycle length for the new proliferating cell}
\State $\underbar P\gets (\Theta_1,\dots,\Theta_{P},\vartheta)$
\Comment{Add the cell cycle length to the list of remaining times in the proliferating list}\\
\ElsIf{$t_{R_D} < \tau$} \Comment{\textbf{A proliferating cell completes cell cycle}}
\State $t \gets t + t_{R_D}$
\Comment{Advance system time}
\State Get index $i$ of $t_{R_D}$ in $\underbar P$
\Comment{Find the proliferating cell with the least remaining time}
\State $\underbar P.delete(\underbar P[i])$ 
\Comment{The proliferating cell with the least remaining time leaves cell cycle}
\State $\underbar P\gets (\Theta_1-t_{R_D},\dots,\Theta_{P}-t_{R_D})$
\Comment{Decrease remaining times of proliferating cells}
\State $Q\gets Q+2$
\Comment{Two new (motile) daughter cell is created}\\
\EndIf
\EndWhile
\end{algorithmic}
\end{algorithm}

\end{document}